\documentclass[11pt]{article}
\usepackage[letterpaper,margin=1in]{geometry}
\usepackage[T1]{fontenc}
\usepackage{lmodern,amsmath,amssymb,amsthm,mathtools,microtype}
\usepackage[hidelinks]{hyperref}
\usepackage{enumitem,needspace}
\setlist[itemize]{leftmargin=*,itemsep=3pt,topsep=5pt}
\setlist[enumerate]{leftmargin=*,itemsep=3pt,topsep=5pt}
\usepackage{todonotes}

\newtheorem{theorem}{Theorem}[section]
\newtheorem{lemma}[theorem]{Lemma}

\newtheorem{corollary}[theorem]{Corollary}
\theoremstyle{definition}

\DeclareMathOperator{\Sym}{Sym}
\DeclareMathOperator{\cl}{cl}
\DeclareMathOperator{\OPT}{OPT}
\DeclareMathOperator{\WFA}{WFA}
\DeclareMathOperator{\sgn}{sgn}
\newcommand{\K}{\mathbb K}
\newcommand{\F}{\mathbb F}
\newcommand{\R}{\mathbb R}
\newcommand{\val}{\nu}
\newcommand{\qs}{\mathcal Q}

\newcommand{\extcost}{\mathrm{ExtCost}}

\title{The $k$-server conjecture is true}
\author{Christian Coester\qquad Elias Koutsoupias\qquad Marek Zbysiński\\[4pt]University of Oxford}
\date{\today}
\begin{document}
\maketitle

\begin{abstract}
The $k$-server conjecture states that a deterministic online algorithm can achieve competitive ratio $k$ on every metric space. We prove the conjecture. Specifically, we show that the work function algorithm satisfies it.

Our proof uses a natural algebraic representation of the work function as a matrix, which encodes all feasible paths to reach a configuration. In this representation, the minimum and addition operations arising in the definition of optimal costs correspond to addition and multiplication of formal expressions, and each work function value corresponds to the determinant of $k$ columns of the matrix. A request arrival updates the representation via a change of basis and row replacement. The amortized analysis is based on a potential function defined in terms of a larger matrix whose coordinates are pairs of coordinates of the original matrix representation.
\end{abstract}

\section{Introduction}\label{sec:introduction}
The $k$-server problem, introduced by Manasse, McGeoch and Sleator~\cite{ManasseMS88}, is a central problem in competitive analysis. It has been repeatedly referred to as the ``holy grail'' of the field \cite{BuchbinderN09,BansalBN10,DehghaniEHLS17,Lee18,CoesterKL21,CoesterK21,FreiKSW25,CoesterU25,GehnenKN25,BrilliantovBA26}, and many techniques originally developed for the $k$-server problem have gone on to impact various other problems in online algorithms. The problem's definition is simple: There are $k$ servers located at points of a metric space. At each time step, a request arrives at a point of the metric space. An online algorithm must serve the request immediately by moving a server to the requested location, without knowledge of future requests. The goal is to minimize the total distance traveled by servers.

When introducing the problem, \cite{ManasseMS88} showed that the competitive ratio of any deterministic online algorithm is at least $k$ on every metric space with $n>k$ points, showed that this is tight when $k=2$ or $n=k+1$, and conjectured that it is tight in general. This question became known as the \emph{$k$-server conjecture}, and has been a driving force for progress in competitive analysis. The first algorithm with a finite competitive ratio on general metrics was obtained in~\cite{FiatRR90}, achieving a competitive ratio exponential in $k$. This was improved to a smaller exponential using the (randomized) Harmonic algorithm in~\cite{Grove91} and~\cite{BartalG00}, which can be derandomized using the techniques of~\cite{Ben-DavidBKTW94}. A substantial improvement was achieved by Koutsoupias and Papadimitriou~\cite{KoutsoupiasP95}, who showed that the \emph{work function algorithm (WFA)} is $(2k-1)$-competitive on general metrics. Another proof of this result was given in~\cite{Koutsoupias99}.

The WFA was introduced in~\cite{ChrobakL92}, and proposed independently also by Karloff and by McGeoch and Sleator. Its definition is simple and generic, and extends far beyond the $k$-server problem: For a configuration $X$ (e.g., the positions of $k$ servers in the $k$-server problem) and time $t$, denote by $w_t(X)$ the minimal (offline) cost to serve the requests up to time $t$ and then end up at configuration $X$. The function $w_t$ is called the work function, and its values can be computed online at time $t$. WFA moves at each time to the configuration that minimizes the sum of its work function value plus the cost (distance) of moving there from the previous configuration.

The conjectured bound of $k$ was established for WFA in many special cases, including $k=2$~\cite{ChrobakL92}, $n=k+1$ (folklore, cf.~\cite{Koutsoupias09}), $n=k+2$~\cite{KoutsoupiasP96}, line metrics~\cite{BartalK04}, weighted stars~\cite{BartalK04} (equivalent to the weighted paging problem), multirays~\cite{CoesterK21}, $k=3$ in the Manhattan plane~\cite{BeinCL02}, $k=3$ on trees~\cite{CoesterK21}, and $k=3$ on circles~\cite{HuangZ24}. Additionally, the Double Coverage algorithm achieves the conjectured bound of $k$ on the line~\cite{ChrobakKPV91} and, more generally, trees~\cite{ChrobakL91}. We show that WFA is $k$-competitive on general metric spaces, thereby confirming the $k$-server conjecture.

\subsection{Preliminaries and our result}
The $k$-server problem is defined on a metric space $(\mathcal M,d)$. There are $k$ servers starting at locations $s_1,\dots,s_k\in\mathcal M$, and a sequence of requests $r_1,\dots,r_T\in\mathcal M$ is revealed online. For convenience, we write $M:=\{s_1,\dots,s_k,r_1,\dots,r_T\}$ for the corresponding set of $k+T$ distinct location labels. Then $d$ induces a pseudometric on $M$; the reason it may only be a pseudometric rather than a metric is that $M$ can contain distinct labels for the same point at distance zero.

A configuration is a $k$-element subset of $M$, representing the locations of $k$ servers. For configurations $X=\{x_1,\ldots,x_k\}$ and $Y=\{y_1,\ldots,y_k\}$, their distance is
\begin{equation}\label{eq:configuration-distance}
D(X,Y)=\min_{\sigma\in S_k}\sum_{i=1}^k d(x_i,y_{\sigma(i)}).
\end{equation}
This is the minimum cost of moving the servers from $X$ to $Y$.

Denote the initial configuration by $C_0=\{s_1,\ldots,s_k\}$. At each time $t$, the algorithm must serve the current request $r_t$ by moving to a configuration $C_t\ni r_t$, paying movement cost $D(C_{t-1},C_t)$. The goal is to minimize the total cost.

The work function is the function $w_t$ that maps a configuration $X$ to the minimum (offline) cost to serve the first $t$ requests and then move to configuration $X$. 
It is well-known (see, e.g., \cite{Koutsoupias09}) that the work function can be computed online by dynamic programming, via $w_0(X)=D(C_0,X)$ and the recurrence
\begin{equation}\label{eq:recurrence}
w_{t}(X)=
\begin{cases}
w_{t-1}(X),&r_t\in X,\\[2pt]
\displaystyle\min_{x\in X}\bigl\{w_{t-1}(X-x+r_t)+d(r_t,x)\bigr\},&r_t\notin X.
\end{cases}
\end{equation}
Thus, the total optimal offline cost can be written as $\OPT_{C_0}(r_1,\ldots,r_T)=\min_X w_T(X)$.

The work function algorithm (WFA) serves the request at time $t$ in a configuration $C_t\ni r_t$ minimizing $w_t(C_t)+D(C_{t-1},C_t)$. For simplicity, we assume here a lazy implementation with $C_t\subseteq C_{t-1}\cup\{r_t\}$, although our result extends to any other tie-breaking rule. We denote the cost of WFA on a request sequence $\sigma$ by $\WFA_{C_0}(\sigma)$.

The following standard accounting argument is the extended-cost method of Chrobak and Larmore~\cite{ChrobakL92}: Instead of bounding directly the cost $D(C_{t-1}, C_t)$ of WFA, we bound the so-called extended cost $\max_X(w_t(X)-w_{t-1}(X))$. We include its proof in the appendix for completeness.

\begin{lemma}[Extended-cost accounting~\cite{ChrobakL92}]\label{lem:accounting}
Let $C_0,C_1,\ldots,C_T$ be any WFA trajectory, and let $\extcost_t=\max_X(w_t(X)-w_{t-1}(X))$. Then
\begin{equation}\label{eq:accounting}
\WFA_{C_0}(r_1,\ldots,r_T)+w_T(C_T)\le\sum_{t=1}^T \extcost_t.
\end{equation}
\end{lemma}

For a configuration $X=\{x_1,\ldots,x_k\}$, define its clique weight by
\[
\cl(X)=\sum_{1\le i<j\le k}d(x_i,x_j).
\]

The following theorem establishes the $k$-server conjecture.

\begin{theorem}\label{thm:main}
Let $(\mathcal M,d)$ be any metric space, and let $C_0=\{s_1,\ldots,s_k\}$ be the initial configuration. For every finite request sequence $\sigma$, the work function algorithm satisfies
\[
\WFA_{C_0}(\sigma)\le k\,\OPT_{C_0}(\sigma)+\cl(C_0).
\]
\end{theorem}

\subsection{Overview of the proof}

By Lemma~\ref{lem:accounting}, it suffices to prove
\begin{equation*}
\sum_{t=1}^T \extcost_t
\le(k+1)w_T(X_T)+\cl(C_0)-\cl(X_T),
\end{equation*}
where $X_T$ is the final optimal offline configuration satisfying $\OPT_{C_0}(\sigma)=w_T(X_T)$.
To prove this bound, it is enough to prove that there exists a potential function $\Psi_t$ which satisfies
\begin{equation}\label{eq:intro-three}
\Psi_0=-\cl(C_0),\qquad
\extcost_t\le\Psi_t-\Psi_{t-1},\qquad
\Psi_t\le(k+1)w_t(X)-\cl(X)
\end{equation}
for every configuration $X$. In order to represent the potential, we need a more complex algebraic structure than the numerical work function. 

We introduce a novel approach which has been standard in the theory of valuated matroids, but has not been previously used for proving the competitiveness of the WFA. We use the fraction field $\K$ which extends the set of formal sums of the form $\sum_{a\in A} f_a z^a$, where $A\subset\R$ is finite. We can define the valuation $\val$ of a formal sum to be the smallest exponent of the free variable $z$ with a nonzero coefficient. For example $\nu(7z^3+4z^6)=3$. The definition extends to the fraction field $\K$ by $\val(f/g)=\val(f)-\val(g)$. This function satisfies the following properties:
\begin{equation*}
\val(ab)=\val(a)+\val(b),\qquad
\val(a+b)\ge\min\{\val(a),\val(b)\}.
\end{equation*}
Since the work function recurrence is defined by taking minimums of sums, using $\K$ and $\val$ are natural candidates for alternative expression of the work function involving ordinary multiplication and addition instead.

To express the current state of the work function, we maintain vectors $q_x\in \K^k$ for every point $x$. In Section~\ref{sec:lift}, we prove that the current work function $w$ can be represented using $q_x$:
\begin{equation*}
w(\{x_1,\dots,x_k\})= \nu(\det(q_{x_1},\ldots,q_{x_k}))
\end{equation*}
Upon arrival of a request $r$, we update the representation by performing the following operations:
\begin{enumerate}
\item Change the vector space basis so that $q_r=e_1$ is a base vector and the determinants are preserved.
\item Replace the first coordinate of every $q_x$ in the new basis by a value proportional to $z^{d(r,x)}$ with valuation $d(r,x)$, with a coefficient that prevents the value to be cancelled out by the determinant. This is achieved by including sufficiently many independent coefficient variables in the field.
\end{enumerate}
We can now consider the problem only in its algebraic form as the work function and extended cost can be represented using only the vectors $q_x$ and their transformations. We use a potential based on the quadratic product $q_x q_y$, which is a $\binom{k+1}{2}$-dimensional vector with coordinates indexed by pairs $(i,j)$ with $i\le j$, whose entries are obtained by taking products of the entries of $q_x$ and $q_y$ in the respective coordinates. The potential is $\Psi=\mu(\qs)$ where $\qs=[\,z^{-d(x,y)}q_xq_y\,]_{x\le y}$ and $\mu(G)=\min_{1\le i_1<\cdots<i_d\le m}
\val(\det(g_{i_1},\ldots,g_{i_d}))$. The $z^{-d(x,y)}$ term and the $\mu$ function draw parallel to previously known potentials for proving the competitiveness of certain cases of the WFA. The other potentials have also used the selection of a subset of points to minimize an expression which contains the work function values, here expressed indirectly using $q_x$, and the distance pairs of used points with negative sign, here expressed using $z^{-d(x,y)}$. 

The check that the potential satisfies $\Psi_t\le(k+1)w_t(X)-\cl(X)$ for any $X$ and $\Psi_0=-\cl(C_0)$ is a relatively straightforward computation shown in Lemmas~\ref{lem:terminal-bound} and~\ref{lem:initial-bound} respectively. 

The heart of the proof is showing that the potential satisfies $\extcost_t\le\Psi_t-\Psi_{t-1}$. In Section~\ref{sec:step}, we prove this by introducing inequalities regarding the current potential, the new potential and the extended cost. They all follow from the definition of the algebraic handling of a request $r$.

\subsection{Related work}\label{sec:related}
\paragraph{Related problems.} The $k$-server problem is closely related to a number of other online problems. For example, the paging (aka caching) problem corresponds to the special case where the metric space is a uniform metric, and was foundational for the field of competitive analysis~\cite{SleatorT85}. The $k$-server problem further belongs to the class of metrical task systems~\cite{BorodinLS92,BubeckCLL21,CoesterL22,BubeckCR23}, and $k$-server on a $(k+1)$-point metric is equivalent to metrical service systems~\cite{ChrobakL92,BubeckCR23}, with further close connections to layered graph traversal~\cite{FiatFKRRV98,BubeckCR25,CoesterT26}. Generalizations of the $k$-server problem include the weighted and the generalized $k$-server problem~\cite{FiatR94,SittersS06,BansalEK17,BansalEKN23,BijoyMC26}, the infinite server problem~\cite{CoesterKL21,BienkowskiBCJ20}, and the $k$-taxi problem~\cite{CoesterK19,GuptaKP24,CoesterP26}.

\paragraph{Work functions and WFA.} The WFA (and a parametrized variant where the work function term is scaled by a constant factor) have been applied to many other problems, including metrical task systems~\cite{BorodinE98}, layered graph traversal~\cite{Burley96}, list update~\cite{AndersonHKRS02}, matching~\cite{KoutsoupiasN03}, generalized $k$-server \cite{Sitters14}. Moreover, work functions have been instrumental to advancements on convex body chasing~\cite{Sellke20,ArgueGTG21}. Variants of WFA have further been proposed as candidates for addressing the $k$-taxi problem~\cite{CoesterK19} and the dynamic optimality conjecture~\cite{Sitters14}.

\paragraph{Quasiconvexity, valuated matroids, and tropical geometry.} An important property of $k$-server work functions is \emph{quasiconvexity}, which was introduced in the online algorithms context by Koutsoupias and Papadimitriou~\cite{KoutsoupiasP95}. However, equivalent or closely related notions have appeared independently across several fields under different terminology. In fact, quasiconvex work functions in the $k$-server setting are mathematically equivalent to \emph{valuated matroids}, a concept introduced by Dress and Wenzel~\cite{DressW1990} in connection with a variant of the greedy algorithm. In economics, Kelso and Crawford~\cite{KelsoC82} introduced the closely related \emph{gross substitutes} condition as a sufficient condition for the existence of Walrasian equilibria. These structures also play a central role in discrete optimization through their equivalence to $M^\natural$-concave functions~\cite{Murota2003}.

The foundational paper by Dress and Wenzel on valuated matroids~\cite{DressW1992} shows that valuated matroids arise naturally from valuations of determinants and Pl\"ucker relations. This relationship has since been further expanded by advances in tropical geometry (see for example, Speyer and Sturmfels~\cite{SpeyerS2004} and Maclagan and Sturmfels~\cite{MaclaganS2015}).

\paragraph{Randomized $k$-server.} A separate direction for the $k$-server problem concerns randomized algorithms (against oblivious adversaries): In this setting, a competitive ratio of $O(\log k)$ is known to be achievable on uniform metrics~\cite{FiatKLMSY91} and weighted stars~\cite{BansalBN12}, matching an $\Omega(\log k)$ lower bound that holds on every metric of $n>k$ points~\cite{BubeckCR23}. The \emph{randomized $k$-server conjecture}, which states that a $O(\log k)$-competitive randomized algorithm exists for all metric spaces, was refuted in~\cite{BubeckCR23} by an $\Omega(\log^2 k)$ lower bound that holds on some metric spaces, even when $n=k+1$. A first polylog$(k,n)$-competitive algorithm for arbitrary $n$-point metrics was achieved in~\cite{BansalBMN15}, and improved to $O(\log n\log^2 k)$ and $O(\log\Delta\log^3 k)$ in~\cite{BubeckCLLM18}, where $\Delta$ is the ratio between the largest and smallest non-zero distance of the metric space. Since both $n$ and $\Delta$ can be infinite, it remains unknown whether randomization provides an advantage on general metric spaces, and the best known competitive ratio even when randomization is allowed is achieved by WFA.


\section{From min-plus recurrences to determinant columns}\label{sec:lift}
This section develops the algebraic state from first principles. Throughout, $z$ is a formal symbol and its exponents represent costs. Se use additional variables as coefficients to distinguish combinatorial choices; their values are never specialized to numbers. We give configurations an ordering whenever they are used to index determinant columns.

\subsection{Formal expressions and their lowest exponents}
For each initial server index $i\in[k]$ and location label $x\in M$, introduce an independent variable $\gamma_{i,x}$. For each time $t\in[T]$ and label $x\in M$, introduce an independent variable $\xi_{t,x}$. Let $\F$ be the field of rational functions over $\mathbb Q$ in all these variables.

Consider finite formal sums
\[
f=\sum_{a\in A} f_a z^a,\qquad A\subseteq\R\text{ finite},\quad f_a\in\F.
\]
Combine equal exponents when adding, and multiply using $z^a z^b=z^{a+b}$. If $f\ne0$, define
\[
\val(f)=\min\{a:f_a\ne0\};\qquad \val(0)=\infty.
\]
For instance, $\val(3z^2-7z^5)=2$ and $\val(z^{-1}+z^4)=-1$. Every nonzero element of the coefficient field $\F$ has valuation zero.

The least-exponent term in a product of two nonzero sums is the product of their least-exponent terms. Its coefficient is nonzero because $\F$ is a field. Therefore these finite sums have no zero divisors. We may form their fraction field $\K$ and extend the definition by
\[
\val(f/g)=\val(f)-\val(g)\qquad(fg\ne0).
\]
This is well defined: if $f/g=f'/g'$, then $fg'=f'g$, and taking least exponents of these products gives the same difference. The two rules we will use are
\begin{equation}\label{eq:valuation}
\val(ab)=\val(a)+\val(b),\qquad
\val(a+b)\ge\min\{\val(a),\val(b)\}.
\end{equation}
The second inequality can be strict. For example, $z^2+(-z^2+z^5)=z^5$. It is precisely this cancellation issue that requires care when replacing a min-plus recurrence by an algebraic one.

\begin{lemma}[Fresh coefficients prevent cancellation]\label{lem:fresh}
Suppose $f_1,\ldots,f_m\in\K$ use none of the independent coefficient variables $\eta_1,\ldots,\eta_m$. For real numbers $a_1,\ldots,a_m$, if at least one $f_i$ is nonzero, then
\[
\val\left(\sum_{i=1}^m\eta_i z^{a_i}f_i\right)
=\min_i\bigl(a_i+\val(f_i)\bigr).
\]
\end{lemma}
\begin{proof}
Put the $f_i$ over a common nonzero denominator that uses none of the $\eta_i$. After this multiplication, all terms are finite sums. At the smallest exponent occurring on the right, the coefficient on the left is a nonempty sum $\sum_{i\in I}\eta_i b_i$, with nonzero $b_i$ independent of all the $\eta_j$. It cannot be zero: after clearing coefficient denominators it is a nonzero polynomial, linear in distinct independent variables $\eta_i$. Dividing by the common denominator gives the assertion.
\end{proof}

\subsection{The initial matching as a determinant}\label{sec:initial-columns}
For every label $x\in M$, including the initial labels themselves, define a column $q_x\in\K^k$ by
\begin{equation}\label{eq:initial-columns}
(q_x)_i=\gamma_{i,x}z^{d(s_i,x)},\qquad i\in[k].
\end{equation}
Thus one uniform formula defines all columns. If $X=(x_1,\ldots,x_k)$ has distinct labels, put
\[
u_0(X)=\det(q_{x_1},\ldots,q_{x_k}).
\]
The determinant expansion is
\begin{equation}\label{eq:matching-det}
u_0(X)=\sum_{\sigma\in S_k}\sgn(\sigma)
\left(\prod_{i=1}^k\gamma_{i,x_{\sigma(i)}}\right)
z^{\sum_{i=1}^k d(s_i,x_{\sigma(i)})}.
\end{equation}
The exponent of each term is the cost of a matching from $C_0$ to $X$. Different permutations give different coefficient monomials, even when their matching costs are equal. Hence the least-cost terms cannot cancel, and
\begin{equation}\label{eq:initial-lift}
\val(u_0(X))=D(C_0,X)=w_0(X).
\end{equation}
In particular, every $k$ distinct columns form a basis.

For two servers, this construction reads
\[
u_0(x,y)
=\gamma_{1,x}\gamma_{2,y}z^{d(s_1,x)+d(s_2,y)}
-\gamma_{1,y}\gamma_{2,x}z^{d(s_1,y)+d(s_2,x)}.
\]
The two possible assignments have opposite signs. If their costs tie, their coefficients still differ, so the minimum cost survives. Using independent symbols rather than unit coefficients is essential for this simple argument.

\subsection{A request as a row replacement}\label{sec:row-update}
We now maintain columns $q_x\in\K^k$ for which
\begin{equation}\label{eq:lift-invariant}
u(X):=\det(q_{x_1},\ldots,q_{x_k}),\qquad \val(u(X))=w(X)
\end{equation}
for every ordered configuration of distinct labels. Exchanging two labels changes the sign of $u(X)$, but leaves its valuation unchanged, as required for an unordered configuration.

Assume $k\ge2$ and a request arrives at $r$. The vector $q_r$ is nonzero, since it belongs to a basis of $k$ distinct columns. We first change coordinates so that $q_r=e_1$, while preserving every determinant exactly.

\begin{lemma}[Determinant-preserving normalization]\label{lem:normalization}
For every nonzero $v\in\K^k$, with $k\ge2$, there is an invertible matrix $B$ satisfying $Bv=e_1$ and $\det B=1$.
\end{lemma}
\begin{proof}
Gaussian elimination gives an invertible matrix $T$ with $Tv=e_1$. Write $c=\det T\ne0$. Multiply the second row of $T$ by $c^{-1}$. The resulting matrix $B$ has determinant one, and still sends $v$ to $e_1$, whose second coordinate is zero.
\end{proof}

Apply this matrix to every old column. Since $\det B=1$, every $k$-column determinant is unchanged. We continue to call the transformed columns $q_x$. The matrix $B$ is chosen using only the old columns and therefore introduces no fresh coefficient variables.

Write $q_x=(a_x,\widehat q_x)$, with $a_x\in\K$ and $\widehat q_x\in\K^{k-1}$. In these coordinates, $q_r=(1,0)$. For request number $t$, set
\begin{equation}\label{eq:delta}
\delta_r=1,\qquad \delta_x=\xi_{t,x}z^{d(r,x)}\quad(x\ne r),
\end{equation}
and replace the columns by
\begin{equation}\label{eq:update}
q'_x=(\delta_x,\widehat q_x).
\end{equation}
Equivalently, replace the first row of the entire column matrix by $(\delta_x)_{x\in M}$ and leave the other rows unchanged. Notice that $\val(\delta_x)=d(r,x)$ for every $x$.

\begin{lemma}[Exact lifted recurrence]\label{lem:lift-update}
The new columns satisfy $\val(\det(q'_{x_1},\ldots,q'_{x_k}))=w'(X)$ for every configuration $X$.
\end{lemma}
\begin{proof}
Expand the new determinant along its first row. The signed minor accompanying $\delta_{x_i}$ is the determinant of the old columns with column $i$ replaced by $e_1=q_r$. Indeed, expanding that replacement determinant along its $e_1$ column gives the same signed minor. Thus
\begin{equation}\label{eq:lift-recurrence}
u'(x_1,\ldots,x_k)
=\sum_{i=1}^k\delta_{x_i}
u(x_1,\ldots,x_{i-1},r,x_{i+1},\ldots,x_k).
\end{equation}
The first-row signs are already included in the replacement determinants; no extra signs are missing from this formula.

If $r\in X$, the summand replacing $r$ itself is $u(X)$, because $\delta_r=1$. Every other summand has two copies of the column $q_r$ and is zero. Hence $u'(X)=u(X)$ exactly.

If $r\notin X$, the factors $\xi_{t,x_i}$ are distinct fresh variables, and the old determinants contain none of them. Lemma~\ref{lem:fresh} gives
\[
\val(u'(X))
=\min_i\bigl\{d(r,x_i)+\val(u(X-x_i+r))\bigr\}
=w'(X),
\]
where the last equality is~\eqref{eq:recurrence}. This proves the invariant after the request.
\end{proof}

After the request column becomes $e_1$, determinants containing that column depend only on the other columns' last $k-1$ coordinates. The row replacement combines exactly those determinants, with the appropriate costs for moving from the request to the final locations. Alternation eliminates configurations in which the request would occur twice.

For example, with two servers the old columns have the form $q_x=(a_x,b_x)$ and $q_r=(1,0)$. The new two-column determinant is
\[
\det(q'_x,q'_y)=\delta_xb_y-\delta_yb_x
=\delta_x\det(q_r,q_y)+\delta_y\det(q_x,q_r).
\]
This is the two-term work function recurrence before taking valuations.

\begin{theorem}[Maintained algebraic state]\label{thm:state}
For every finite request sequence, the initialization~\eqref{eq:initial-columns} and the determinant-preserving normalization followed by~\eqref{eq:update} produce a $k\times n$ matrix whose $k$-column determinants have valuations exactly equal to the corresponding work function values at every time.
\end{theorem}
\begin{proof}
The initialization is~\eqref{eq:initial-lift}; preservation is Lemma~\ref{lem:lift-update}. All coordinate operations before time $t$ use only coefficient variables from earlier times, so the required freshness holds inductively. Although individual matrix entries may become rational functions, the determinants satisfy the exact recurrence~\eqref{eq:lift-recurrence}; in particular, their dependence on the coefficient variables is controlled by this recurrence.
\end{proof}

\section{Quadratic products and their determinants}\label{sec:quadratic}
We next explain the space in which the potential lives. This construction uses ordinary polynomials in coordinate symbols; it is separate from the formal cost variable $z$.

\subsection{The symmetric square in coordinates}
Identify a vector $a=(a_1,\ldots,a_k)\in\K^k$ with the linear polynomial $a_1e_1+\cdots+a_ke_k$, where $e_1,\ldots,e_k$ are commuting formal symbols. For vectors $a,b$, their product $ab$ is the usual product of these polynomials. It is a homogeneous polynomial of degree two. Its coordinate at $e_i^2$ is $a_ib_i$, and its coordinate at $e_ie_j$, for $i<j$, is $a_ib_j+a_jb_i$.

The vector space of these quadratic polynomials is denoted $\Sym^2\K^k$. Its dimension is $N=k(k+1)/2$, and its monomial basis, in lexicographic order, is
\begin{equation}\label{eq:monomial-basis}
(e_ie_j)_{1\le i\le j\le k}.
\end{equation}
The notation names this concrete vector space; no tensor-product construction is required. For $k=2$, the product is simply the three-dimensional column
\[
ab=\begin{pmatrix}a_1b_1\\a_1b_2+a_2b_1\\a_2b_2\end{pmatrix}.
\]
For $k=3$, the six coordinate symbols are $e_1^2,e_1e_2,e_1e_3,e_2^2,e_2e_3,e_3^2$.

An invertible matrix $A$ acts on linear polynomials by sending each basis vector $e_i$ to $Ae_i$. Substituting these linear polynomials into a quadratic polynomial gives an invertible linear map on the quadratic space, denoted $\Sym^2 A$. By construction,
\[
(\Sym^2 A)(ab)=(Aa)(Ab).
\]
Consequently, if $a_1,\ldots,a_k$ are a basis of $\K^k$, their $N$ products $a_ia_j$, $i\le j$, are a basis of the quadratic space.

The following lemma is a well-known result that goes back to~\cite{schlaefli1851}, as it is reported in~\cite[pp. 52-53]{muir1960}. For completeness, we add its proof in Appendix~\ref{apdx:sym2proof}.
\begin{lemma}[Determinant of the symmetric square]\label{lem:symdet}
For every invertible $k\times k$ matrix $A$,
\begin{equation}\label{eq:symdet}
\det(\Sym^2 A)=(\det A)^{k+1}.
\end{equation}
\end{lemma}

The exponent $k+1$ has a simple counting interpretation: a given basis vector occurs twice in its square and once in each of its $k-1$ mixed products. This is why the quadratic space is suited to the $(k+1)$ coefficient in the extended-cost bound.

\subsection{Minimum valuations of maximal minors}\label{sec:minor-values}
Let $G=[g_1\ \cdots\ g_m]$ be a matrix whose columns span $\K^d$. Define
\begin{equation}\label{eq:mu}
\mu(G)=\min_{1\le i_1<\cdots<i_d\le m}
\val(\det(g_{i_1},\ldots,g_{i_d})).
\end{equation}
A dependent choice contributes $+\infty$. Because the family spans, at least one choice is independent. Column order affects only the determinant's sign, which has valuation zero.

\Needspace{12\baselineskip}
\begin{lemma}[Rules for minor valuations]\label{lem:minor-rules}
The following facts hold for spanning column families.
\begin{enumerate}
\item If $A$ is an invertible $d\times d$ matrix, then $\mu(AG)=\val(\det A)+\mu(G)$.
\item If every column of a spanning family $F$ is a linear combination of columns of $G$ with coefficients of nonnegative valuation, then $\mu(F)\ge\mu(G)$.
\item If $G$ is block diagonal
with blocks $G_1,\ldots,G_s$, then $\mu(G)=\sum_{j=1}^s\mu(G_j)$.
\end{enumerate}
\end{lemma}
\begin{proof}
For the first statement, every selected determinant is multiplied by $\det A$.
For the second, expand a determinant of $d$ columns of $F$ multilinearly in columns of $G$. Each summand is a determinant of $d$ columns of $G$, multiplied by $d$ coefficients of total valuation at least zero. Every summand therefore has valuation at least $\mu(G)$, and so does their sum by~\eqref{eq:valuation}.

For the third statement, let $d_j$ be the number of rows of $G_j$. A nonzero full determinant must choose exactly $d_j$ columns from each block: choosing more makes those columns dependent, and choosing fewer leaves that block unspanned. The determinant is then the product of the block determinants, and its valuation is the sum of their valuations. Choosing a minimum determinant independently in every block attains the claimed sum.
\end{proof}

We will also use block upper triangular matrices. For a matrix $\left(\begin{smallmatrix}U&*\\0&W\end{smallmatrix}\right)$ with square diagonal blocks, its determinant is $\det U\det W$, so the determinant valuation is the sum of the two block valuations.

\section{The potential and its increase at a request}\label{sec:step}
Fix any total order on the finite label set $M$. At each time define the $N\times\binom{n+1}{2}$ matrix $\qs$, where $N=\binom{k+1}{2}$, and potential $\Psi$ by
\begin{equation}\label{eq:Q-definition}
\qs=[\,z^{-d(x,y)}q_xq_y\,]_{x\le y},\qquad
\Psi=\mu(\qs).
\end{equation}
The columns are written in the monomial coordinates~\eqref{eq:monomial-basis}. This family spans: choose any configuration of $k$ distinct labels, whose columns form a basis by Theorem~\ref{thm:state}; their pairwise products form a basis of the quadratic space. Multiplying those products by nonzero distance weights preserves their independence.

Before processing a request, we apply the matrix $B$ from Lemma~\ref{lem:normalization} to every column. This replaces $\qs$ by $(\Sym^2 B)\qs$. Lemmas~\ref{lem:symdet} and~\ref{lem:minor-rules} show that $\Psi$ is unchanged, since $\det B=1$. Therefore we may analyze the request in coordinates where $q_r=e_1$.

\subsection{Three quantities associated with the request}
Write $q_x=(a_x,\widehat q_x)$ and let $\pi:\K^k\to\K^{k-1}$ delete the first coordinate. Define the weighted columns and their projection by
\begin{equation}\label{eq:H}
h_x=z^{-d(r,x)}q_x,\qquad H=[h_x]_{x\in M},\qquad
\widehat H=[\pi h_x]_{x\in M}.
\end{equation}
Let $\widehat{\qs}$ be the bottom coordinate block of $\qs$, consisting of the rows $e_ie_j$ with $2\le i\le j\le k$:
\begin{equation}\label{eq:A}
\widehat{\qs}=[\,z^{-d(x,y)}\widehat q_x\widehat q_y\,]_{x\le y}.
\end{equation}
Consider the three quantities $\mu(H), \mu(\widehat H)$ and $\mu(\widehat{\qs})$. The determinants defining these quantities have sizes $k$, $k-1$, and $\binom{k}{2}=N-k$, respectively. Each family spans its coordinate space: $H$ spans $\K^k$, its projection spans $\K^{k-1}$, and products of a projected basis span the remaining quadratic coordinates. All three quantities are therefore finite.

We will prove
\begin{equation}\label{eq:three-comparisons}
\Psi\le \mu(\widehat{\qs})+\mu(H),\qquad
\Psi'\ge \mu(\widehat{\qs})+\mu(\widehat H),\qquad
w'(X)-w(X)\le\mu(\widehat H)-\mu(H)\quad\text{for all }X.
\end{equation}
The first inequality selects one useful minor of the old matrix $\qs$. The second bounds every minor of the new matrix $\qs'$. The third relates the extended cost to the difference between these two potential bounds.

\subsection{A useful minor of the old quadratic matrix}
\begin{lemma}\label{lem:old-potential}
The old potential before the request satisfies $\Psi\le \mu(\widehat{\qs})+\mu(H)$.
\end{lemma}
\begin{proof}
The first $k$ quadratic coordinates are
\[
e_1^2,e_1e_2,\ldots,e_1e_k.
\]
Multiplying a vector $h\in\K^k$ by $e_1$ puts its $k$ coordinates into these rows and puts zeros into every remaining quadratic row. Because $q_r=e_1$, the columns
\[
e_1h_x=z^{-d(r,x)}q_rq_x
\]
are actual columns of $\qs$, namely its pairs $(r,x)$.

Choose $k$ columns of $H$ attaining $\mu(H)$ and take the corresponding $k$ columns $e_1h_x$ of $\qs$. Their first $k$ rows give the selected $H$ minor and their bottom rows are zero. Next choose $N-k$ columns of $\qs$ whose bottom blocks attain $\mu(\widehat{\qs})$ in $\widehat{\qs}$. They do not include any column of the first selected columns, since those have zero bottom blocks and could not occur in a nonzero bottom-block determinant.

Together these $N$ columns form a square matrix of the form
\[
\begin{pmatrix}H_*&*\\0&\widehat Q_*\end{pmatrix},
\]
where $\val(\det H_*)=\mu(H)$ and $\val(\det\widehat Q_*)=\mu(\widehat{\qs})$. The determinant has valuation $\mu(H)+\mu(\widehat{\qs})$. Since $\Psi$ is the minimum over all maximal minors of $\qs$, it is at most this value.
\end{proof}

\subsection{A bound on every minor of the new quadratic matrix}
\begin{lemma}\label{lem:new-potential}
The new potential after the request satisfies $\Psi'\ge \mu(\widehat{\qs})+\mu(\widehat H)$.
\end{lemma}
\begin{proof}
View $\widehat q_x$ as a linear polynomial in $e_2,\ldots,e_k$. The update is $q'_x=\delta_xe_1+\widehat q_x$. Expanding one new quadratic column gives
\begin{align}
z^{-d(x,y)}q'_xq'_y
={}&z^{-d(x,y)}\delta_x\delta_y e_1^2\notag\\
&+e_1z^{-d(x,y)}(\delta_x\widehat q_y+\delta_y\widehat q_x)\notag\\
&+z^{-d(x,y)}\widehat q_x\widehat q_y.\label{eq:quadratic-expansion}
\end{align}
The expressions in these three lines occupy disjoint coordinate blocks: the single row $e_1^2$; the $k-1$ rows $e_1e_j$, $j\ge2$; and the $N-k$ rows not involving $e_1$.

Compare the new columns to the following family, supported separately on these three blocks:
\begin{equation}\label{eq:block-family}
\{e_1^2\}\;\cup\;
\{e_1\pi h_x:x\in M\}\;\cup\;
\{z^{-d(x,y)}\widehat q_x\widehat q_y:x\le y\}.
\end{equation}
Its minimum full-minor valuation is $0+\mu(\widehat H)+\mu(\widehat{\qs})$, by the block rule of Lemma~\ref{lem:minor-rules}. We show that every column in~\eqref{eq:quadratic-expansion} is a linear combination of this family with coefficient valuations at least zero.

The coefficient of $e_1^2$ has valuation
\begin{equation}\label{eq:triangle-defect}
\val(z^{-d(x,y)}\delta_x\delta_y)
=d(r,x)+d(r,y)-d(x,y)\ge0.
\end{equation}
For a mixed term, use $\pi h_y=z^{-d(r,y)}\widehat q_y$ to write
\[
e_1z^{-d(x,y)}\delta_x\widehat q_y
=\bigl(z^{d(r,y)-d(x,y)}\delta_x\bigr)e_1\pi h_y.
\]
The coefficient in parentheses has non-negative valuation. The other mixed term in~\eqref{eq:quadratic-expansion} has the analogous expression with $x$ and $y$ interchanged. Finally, the last line of~\eqref{eq:quadratic-expansion} is already a column of the third family, with coefficient one.

The second rule of Lemma~\ref{lem:minor-rules} now shows that every maximal minor of $\qs'$ has valuation at least $\mu(\widehat{\qs})+\mu(\widehat H)$. Taking their minimum proves the assertion.
\end{proof}

Combining the preceding two lemmas gives
\begin{equation}\label{eq:potential-increment}
\Psi'-\Psi\ge\mu(\widehat H)-\mu(H).
\end{equation}

\subsection{Why the potential change pays the extended cost}
\begin{lemma}\label{lem:extended-cost}
For every configuration $X$, $w'(X)-w(X)\le\mu(\widehat H)-\mu(H)$.
\end{lemma}
\begin{proof}
First suppose $r\notin X$. The old columns $q_{x_1},\ldots,q_{x_k}$ are a basis, so write the request column in this basis:
\begin{equation}\label{eq:cramer-expansion}
e_1=q_r=\sum_{i=1}^k c_iq_{x_i}.
\end{equation}
Replacing the $i$th column of this basis by $e_1$ multiplies its determinant by $c_i$: substitute~\eqref{eq:cramer-expansion} and use multilinearity, observing that every summand except $i$ has a repeated column and vanishes. This is Cramer's rule, and it gives
\[
c_i=\frac{u(X-x_i+r)}{u(X)},\qquad
\val(c_i)=w(X-x_i+r)-w(X).
\]
By the work function recurrence, the increase at $X$ is therefore
\begin{equation}\label{eq:m}
m:=w'(X)-w(X)=\min_i\bigl\{\val(c_i)+d(r,x_i)\bigr\}.
\end{equation}

Consider any $k-1$ columns $\pi h_{y_1},\ldots,\pi h_{y_{k-1}}$ of $\widehat H$. We have
\begin{align}
\det(\pi h_{y_1},\ldots,\pi h_{y_{k-1}})
&=\det(e_1,h_{y_1},\ldots,h_{y_{k-1}})\notag\\
&=\sum_{i=1}^k c_i z^{d(r,x_i)}
\det(h_{x_i},h_{y_1},\ldots,h_{y_{k-1}}).\label{eq:cramer-minor}
\end{align}
The first equality is expansion along the first column; the second substitutes~\eqref{eq:cramer-expansion} and uses $q_{x_i}=z^{d(r,x_i)}h_{x_i}$.

Every determinant in the final line is either zero or a maximal minor of $H$, up to a column ordering. Its valuation is at least $\mu(H)$. Equation~\eqref{eq:m} bounds the valuation of each coefficient $c_i z^{d(r,x_i)}$ below by $m$. Hence every summand has valuation at least $\mu(H)+m$, and so does their sum. Minimizing the left side over the projected columns proves $\mu(\widehat H)\ge\mu(H)+m$.

For the case $r\in X$, the work function increase is zero. Independently of the choice of $X$, we have $\mu(\widehat H)\ge\mu(H)$: since $h_r=e_1$, each determinant defining $\mu(\widehat H)$ is, by the first equality of~\eqref{eq:cramer-minor}, already a determinant of $k$ columns of $H$. Its valuation is therefore at least $\mu(H)$. This completes the remaining case.
\end{proof}

\begin{corollary}[The potential pays every request]\label{cor:step}
For each time $t$,
\[
\extcost_t\le\Psi_t-\Psi_{t-1}.
\]
\end{corollary}
\begin{proof}
Immediate from Lemma~\ref{lem:extended-cost} and~\eqref{eq:potential-increment}.
\end{proof}

\section{Initial and final potential values}\label{sec:endpoints}
We now prove bounds on the initial and terminal values of the potential function.

\begin{lemma}[Terminal bound]\label{lem:terminal-bound}
At every time and for every configuration $X$, we have
\begin{equation}\label{eq:terminal-bound}
\Psi\le(k+1)w(X)-\cl(X).
\end{equation}
\end{lemma}
\begin{proof}
Let $A$ be the $k\times k$ matrix with columns $q_{x_1},\ldots,q_{x_k}$. Select from $\qs$ the $N$ columns
\[
z^{-d(x_i,x_j)}q_{x_i}q_{x_j},\qquad 1\le i\le j\le k.
\]
Without the scalar weights, their matrix is $\Sym^2 A$, since its columns are the images of the coordinate monomials $e_ie_j$. The determinant is nonzero and, by Lemma~\ref{lem:symdet}, its valuation is $(k+1)\val(\det A)=(k+1)w(X)$. The scalar weights subtract
\[
\sum_{i\le j}d(x_i,x_j)=\sum_{i<j}d(x_i,x_j)=\cl(X)
\]
from the determinant valuation; diagonal distances are zero. This gives one maximal minor of $\qs$ with the claimed valuation, and $\Psi$ is the minimum over all such minors.
\end{proof}

\Needspace{5\baselineskip}
\begin{lemma}[Initial value]\label{lem:initial-bound}
The initialization~\eqref{eq:initial-columns} satisfies $\Psi_0=-\cl(C_0)$.
\end{lemma}
\begin{proof}
The terminal bound applied at time zero with $X=C_0$ gives $\Psi_0\le-\cl(C_0)$, since $w_0(C_0)=0$.

For the reverse inequality, fix a row $e_ie_j$ of the initial quadratic matrix and a column indexed by $(x,y)$. If $i<j$, its entry is
\[
z^{-d(x,y)}\bigl((q_x)_i(q_y)_j+(q_x)_j(q_y)_i\bigr).
\]
The valuation of the first summand is
\[
d(s_i,x)+d(s_j,y)-d(x,y)\ge-d(s_i,s_j)
\]
by triangle inequality along the path $x,s_i,s_j,y$.
The second summand has the same lower bound. Thus the whole entry has valuation at least $-d(s_i,s_j)$.

If $i=j$, the entry has just one product. Its valuation is $d(s_i,x)+d(s_i,y)-d(x,y)\ge0=-d(s_i,s_i)$. Therefore every entry in row $e_ie_j$ has valuation at least $-d(s_i,s_j)$, including diagonal rows.

Expand any $N\times N$ minor by permutations. Each product chooses one entry from each row and hence has valuation at least
\[
-\sum_{i\le j}d(s_i,s_j)=-\cl(C_0).
\]
The valuation of the sum cannot be smaller. This bounds every maximal minor from below, giving $\Psi_0\ge-\cl(C_0)$.
\end{proof}

\begin{theorem}\label{thm:finite}
Let $X_T$ be a configuration minimizing $w_T$. Then
\begin{align}
\sum_{t=1}^T \extcost_t
&\le(k+1)w_T(X_T)+\cl(C_0)-\cl(X_T),\label{eq:finite-E}\\
\WFA_{C_0}(r_1,\ldots,r_T)
&\le k\,w_T(X_T)+\cl(C_0)-\cl(X_T).\label{eq:finite-WFA}
\end{align}
\end{theorem}
\begin{proof}
Sum Corollary~\ref{cor:step} over the requests. The potential telescopes, and the endpoint lemmas give
\[
\sum_t \extcost_t\le\Psi_T-\Psi_0
\le(k+1)w_T(X_T)-\cl(X_T)+\cl(C_0).
\]
Subtract $w_T(C_T)\ge w_T(X_T)$ from the accounting inequality in Lemma~\ref{lem:accounting} to obtain~\eqref{eq:finite-WFA}.
\end{proof}

\section*{Acknowledgments}

This work began with a potential function designed by the authors that proves the previously unknown case of $k=3$ servers on arbitrary metrics. Our initial proof for three servers showed the correctness of this potential function by establishing infeasibility of a large family of linear programs; this proof was obtained without AI assistance. Through discussions with ChatGPT 5.5 Pro and Gemini 3.1 Pro, we gained a deeper understanding of the potential. These discussions led to a more symmetric reformulation of the potential, again designed by the authors, and an alternative proof for $k=3$. Based on this, ChatGPT 6 Astra subsequently derived an algebraic proof of correctness for any $k$. The proof in this paper is an adaptation thereof using an explicit column representation of work functions, which provides a more natural algebraic representation. The authors supplied this representation to ChatGPT 6 Astra, which adapted the previous proof and assisted with drafting some sections of this paper. The authors subsequently revised those drafts.

Verification of some ideas of generalization from $k=3$ to $k=4$ has used the University of Oxford Advanced Research Computing (ARC) facility \cite{richards_2015_22558}.

Christian Coester is funded by the European Union (ERC, CCOO, 101165139). Views and opinions expressed are however those of the author(s) only and do not necessarily reflect those of the European Union or the European Research Council. Neither the European Union nor the granting authority can be held responsible for them. Marek Zbysiński is funded by EPSRC grant EP/Z534870/1.

\begin{quote} 
The second author, Elias Koutsoupias, dedicates this work to his constant friends Amos Fiat, Anna Karlin, and Christos Papadimitriou.
\end{quote}

\bibliographystyle{alpha}
\bibliography{references}

@article{AndersonHKRS02,
  author       = {Eric J. Anderson and
                  Kirsten Hildrum and
                  Anna R. Karlin and
                  April Rasala and
                  Michael E. Saks},
  title        = {On list update and work function algorithms},
  journal      = {Theor. Comput. Sci.},
  volume       = {287},
  number       = {2},
  pages        = {393--418},
  year         = {2002},
  doi          = {10.1016/S0304-3975(01)00253-5}
}

@article{ArgueGTG21,
  author       = {C. J. Argue and
                  Anupam Gupta and
                  Ziye Tang and
                  Guru Guruganesh},
  title        = {Chasing Convex Bodies with Linear Competitive Ratio},
  journal      = {J. {ACM}},
  volume       = {68},
  number       = {5},
  pages        = {32:1--32:10},
  year         = {2021},
  doi          = {10.1145/3450349}
}

@article{BansalBMN15,
  author       = {Nikhil Bansal and
                  Niv Buchbinder and
                  Aleksander Madry and
                  Joseph Naor},
  title        = {A Polylogarithmic-Competitive Algorithm for the \emph{k}-Server Problem},
  journal      = {J. {ACM}},
  volume       = {62},
  number       = {5},
  pages        = {40:1--40:49},
  year         = {2015},
  doi          = {10.1145/2783434}
}

@inproceedings{BansalBN10,
  author       = {Nikhil Bansal and
                  Niv Buchbinder and
                  Joseph Naor},
  title        = {Towards the Randomized k-Server Conjecture: {A} Primal-Dual Approach},
  booktitle    = {Proceedings of the Twenty-First Annual {ACM-SIAM} Symposium on Discrete Algorithms, {SODA}},
  pages        = {40--55},
  year         = {2010},
  doi          = {10.1137/1.9781611973075.5}
}

@article{BansalBN12,
  author       = {Nikhil Bansal and
                  Niv Buchbinder and
                  Joseph Naor},
  title        = {A Primal-Dual Randomized Algorithm for Weighted Paging},
  journal      = {J. {ACM}},
  volume       = {59},
  number       = {4},
  pages        = {19:1--19:24},
  year         = {2012},
  doi          = {10.1145/2339123.2339126}
}

@inproceedings{BansalEK17,
  author       = {Nikhil Bansal and
                  Marek Eli{\'{a}}s and
                  Grigorios Koumoutsos},
  title        = {Weighted k-Server Bounds via Combinatorial Dichotomies},
  booktitle    = {58th {IEEE} Annual Symposium on Foundations of Computer Science, {FOCS}},
  pages        = {493--504},
  year         = {2017},
  doi          = {10.1109/FOCS.2017.52}
}

@article{BansalEKN23,
  author       = {Nikhil Bansal and
                  Marek Eli{\'{a}}s and
                  Grigorios Koumoutsos and
                  Jesper Nederlof},
  title        = {Competitive Algorithms for Generalized \emph{k}-Server in Uniform Metrics},
  journal      = {{ACM} Trans. Algorithms},
  volume       = {19},
  number       = {1},
  pages        = {8:1--8:15},
  year         = {2023},
  doi          = {10.1145/3568677}
}

@article{BartalG00,
  author       = {Yair Bartal and
                  Eddie Grove},
  title        = {The harmonic \emph{k}-server algorithm is competitive},
  journal      = {J. {ACM}},
  volume       = {47},
  number       = {1},
  pages        = {1--15},
  year         = {2000},
  doi          = {10.1145/331605.331606}
}

@article{BartalK04,
  author       = {Yair Bartal and
                  Elias Koutsoupias},
  title        = {On the competitive ratio of the work function algorithm for the k-server problem},
  journal      = {Theor. Comput. Sci.},
  volume       = {324},
  number       = {2-3},
  pages        = {337--345},
  year         = {2004},
  doi          = {10.1016/J.TCS.2004.06.001}
}

@article{BeinCL02,
  author       = {Wolfgang W. Bein and
                  Marek Chrobak and
                  Lawrence L. Larmore},
  title        = {The 3-server problem in the plane},
  journal      = {Theor. Comput. Sci.},
  volume       = {289},
  number       = {1},
  pages        = {335--354},
  year         = {2002},
  doi          = {10.1016/S0304-3975(01)00305-X}
}

@article{Ben-DavidBKTW94,
  author       = {Shai Ben{-}David and
                  Allan Borodin and
                  Richard M. Karp and
                  G{\'{a}}bor Tardos and
                  Avi Wigderson},
  title        = {On the Power of Randomization in On-Line Algorithms},
  journal      = {Algorithmica},
  volume       = {11},
  number       = {1},
  pages        = {2--14},
  year         = {1994},
  doi          = {10.1007/BF01294260}
}

@inproceedings{BienkowskiBCJ20,
  author       = {Marcin Bienkowski and
                  Jaroslaw Byrka and
                  Christian Coester and
                  Lukasz Jez},
  title        = {Unbounded lower bound for k-server against weak adversaries},
  booktitle    = {Proceedings of the 52nd Annual {ACM} {SIGACT} Symposium on Theory of Computing, {STOC}},
  pages        = {1165--1169},
  year         = {2020},
  doi          = {10.1145/3357713.3384306}
}

@book{BorodinE98,
  author       = {Allan Borodin and
                  Ran El{-}Yaniv},
  title        = {Online computation and competitive analysis},
  publisher    = {Cambridge University Press},
  year         = {1998}
}

@inproceedings{BijoyMC26,
  author       = {Adithya Bijoy and
                  Ankit Mondal and
                  Ashish Chiplunkar},
  title        = {Weighted k-Server Admits an Exponentially Competitive Algorithm},
  booktitle    = {Proceedings of the 2026 Annual {ACM-SIAM} Symposium on Discrete Algorithms, {SODA}},
  pages        = {4188--4208},
  year         = {2026},
  doi          = {10.1137/1.9781611978971.154}
}

@article{BorodinLS92,
  author       = {Allan Borodin and
                  Nathan Linial and
                  Michael E. Saks},
  title        = {An Optimal On-Line Algorithm for Metrical Task System},
  journal      = {J. {ACM}},
  volume       = {39},
  number       = {4},
  pages        = {745--763},
  year         = {1992},
  url          = {https://doi.org/10.1145/146585.146588},
  doi          = {10.1145/146585.146588},
  bibsource    = {dblp computer science bibliography, https://dblp.org}
}

@article{BrilliantovBA26,
  author       = {Kirill Brilliantov and
                  Etienne Bamas and
                  Emmanuel Abb{\'{e}}},
  title        = {k-server-bench: Automating Potential Discovery for the k-Server Conjecture},
  journal      = {arXiv preprint},
  year         = {2026},
  doi          = {10.48550/ARXIV.2604.07240}
}

@article{BubeckCLL21,
  author       = {S{\'{e}}bastien Bubeck and
                  Michael B. Cohen and
                  James R. Lee and
                  Yin Tat Lee},
  title        = {Metrical Task Systems on Trees via Mirror Descent and Unfair Gluing},
  journal      = {{SIAM} J. Comput.},
  volume       = {50},
  number       = {3},
  pages        = {909--923},
  year         = {2021},
  doi          = {10.1137/19M1237879}
}

@inproceedings{BubeckCLLM18,
  author       = {S{\'{e}}bastien Bubeck and
                  Michael B. Cohen and
                  Yin Tat Lee and
                  James R. Lee and
                  Aleksander Madry},
  title        = {k-server via multiscale entropic regularization},
  booktitle    = {Proceedings of the 50th Annual {ACM} {SIGACT} Symposium on Theory of Computing, {STOC}},
  pages        = {3--16},
  year         = {2018},
  doi          = {10.1145/3188745.3188798}
}

@article{BubeckCR25,
  author       = {S{\'{e}}bastien Bubeck and
                  Christian Coester and
                  Yuval Rabani},
  title        = {Shortest Paths Without a Map, but with an Entropic Regularizer},
  journal      = {{SIAM} J. Comput.},
  volume       = {54},
  number       = {5},
  pages        = {S22--265},
  year         = {2025},
  doi          = {10.1137/22M1539149}
}

@inproceedings{BubeckCR23,
  author       = {S{\'{e}}bastien Bubeck and
                  Christian Coester and
                  Yuval Rabani},
  title        = {The Randomized k-Server Conjecture Is False!},
  booktitle    = {Proceedings of the 55th Annual {ACM} Symposium on Theory of Computing, {STOC}},
  pages        = {581--594},
  publisher    = {{ACM}},
  year         = {2023},
  doi          = {10.1145/3564246.3585132}
}

@article{BuchbinderN09,
  author       = {Niv Buchbinder and
                  Joseph Naor},
  title        = {The Design of Competitive Online Algorithms via a Primal-Dual Approach},
  journal      = {Found. Trends Theor. Comput. Sci.},
  volume       = {3},
  number       = {2-3},
  pages        = {93--263},
  year         = {2009},
  doi          = {10.1561/0400000024}
}

@article{Burley96,
  author       = {William R. Burley},
  title        = {Traversing Layered Graphs Using the Work Function Algorithm},
  journal      = {J. Algorithms},
  volume       = {20},
  number       = {3},
  pages        = {479--511},
  year         = {1996},
  doi          = {10.1006/JAGM.1996.0024}
}

@article{ChrobakKPV91,
  author       = {Marek Chrobak and
                  Howard J. Karloff and
                  T. H. Payne and
                  Sundar Vishwanathan},
  title        = {New Results on Server Problems},
  journal      = {{SIAM} J. Discret. Math.},
  volume       = {4},
  number       = {2},
  pages        = {172--181},
  year         = {1991},
  doi          = {10.1137/0404017}
}

@article{ChrobakL91,
  author       = {Marek Chrobak and
                  Lawrence L. Larmore},
  title        = {An Optimal On-Line Algorithm for k-Servers on Trees},
  journal      = {{SIAM} J. Comput.},
  volume       = {20},
  number       = {1},
  pages        = {144--148},
  year         = {1991},
  doi          = {10.1137/0220008}
}

@inproceedings{ChrobakL92,
  author       = {Marek Chrobak and
                  Lawrence L. Larmore},
  title        = {The Server Problem and On-Line Games},
  booktitle    = {On-Line Algorithms, Proceedings of a {DIMACS} Workshop},
  series       = {{DIMACS} Series in Discrete Mathematics and Theoretical Computer Science},
  volume       = {7},
  year         = {1992},
  doi          = {10.1090/DIMACS/007/02},
}

@inproceedings{CoesterK19,
  author       = {Christian Coester and
                  Elias Koutsoupias},
  title        = {The online $k$-taxi problem.},
  booktitle    = {Proceedings of the 51st Annual {ACM} {SIGACT} Symposium on Theory of Computing, {STOC}},
  pages        = {1136--1147},
  year         = {2019},
  doi          = {10.1145/3313276.3316370}
}

@inproceedings{CoesterK21,
  author       = {Christian Coester and
                  Elias Koutsoupias},
  title        = {Towards the k-Server Conjecture: {A} Unifying Potential, Pushing the Frontier to the Circle},
  booktitle    = {48th International Colloquium on Automata, Languages, and Programming, {ICALP}},
  volume       = {198},
  pages        = {57:1--57:20},
  year         = {2021},
  doi          = {10.4230/LIPICS.ICALP.2021.57}
}

@article{CoesterKL21,
  author       = {Christian Coester and
                  Elias Koutsoupias and
                  Philip Lazos},
  title        = {The Infinite Server Problem},
  journal      = {{ACM} Trans. Algorithms},
  volume       = {17},
  number       = {3},
  pages        = {20:1--20:23},
  year         = {2021},
  doi          = {10.1145/3456632}
}

@article{CoesterL22,
  author       = {Christian Coester and
                  James R. Lee},
  title        = {Pure Entropic Regularization for Metrical Task Systems},
  journal      = {Theory Comput.},
  volume       = {18},
  pages        = {1--24},
  year         = {2022},
  doi          = {10.4086/TOC.2022.V018A023}
}

@inproceedings{CoesterP26,
  author       = {Christian Coester and
                  Tze{-}Yang Poon},
  title        = {Online 3-Taxi on General Metrics},
  booktitle    = {Proceedings of the 2026 Annual {ACM-SIAM} Symposium on Discrete Algorithms, {SODA}},
  pages        = {6659--6673},
  year         = {2026},
  doi          = {10.1137/1.9781611978971.238}
}

@inproceedings{CoesterT26,
  author       = {Christian Coester and
                  Alexa Tudose},
  title        = {Chasing Small Sets Optimally Against Adaptive Adversaries},
  booktitle    = {53rd International Colloquium on Automata, Languages, and Programming, {ICALP}},
  volume       = {374},
  pages        = {67:1--67:22},
  year         = {2026},
  doi          = {10.4230/LIPICS.ICALP.2026.67}
}

@inproceedings{CoesterU25,
  author       = {Christian Coester and
                  Jack Umenberger},
  title        = {Smoothed Analysis of Online Metric Problems},
  booktitle    = {33rd Annual European Symposium on Algorithms, {ESA}},
  series       = {LIPIcs},
  volume       = {351},
  pages        = {115:1--115:14},
  year         = {2025},
  doi          = {10.4230/LIPICS.ESA.2025.115}
}

@inproceedings{DehghaniEHLS17,
  author       = {Sina Dehghani and
                  Soheil Ehsani and
                  MohammadTaghi Hajiaghayi and
                  Vahid Liaghat and
                  Saeed Seddighin},
  title        = {Stochastic k-Server: How Should Uber Work?},
  booktitle    = {44th International Colloquium on Automata, Languages, and Programming, {ICALP}},
  series       = {LIPIcs},
  volume       = {80},
  pages        = {126:1--126:14},
  year         = {2017},
  doi          = {10.4230/LIPICS.ICALP.2017.126}
}

@article{FiatKLMSY91,
  author       = {Amos Fiat and
                  Richard M. Karp and
                  Michael Luby and
                  Lyle A. McGeoch and
                  Daniel Dominic Sleator and
                  Neal E. Young},
  title        = {Competitive Paging Algorithms},
  journal      = {J. Algorithms},
  volume       = {12},
  number       = {4},
  pages        = {685--699},
  year         = {1991},
  doi          = {10.1016/0196-6774(91)90041-V}
}

@article{FiatR94,
  author       = {Amos Fiat and
                  Moty Ricklin},
  title        = {Competitive Algorithms for the Weighted Server Problem},
  journal      = {Theor. Comput. Sci.},
  volume       = {130},
  number       = {1},
  pages        = {85--99},
  year         = {1994},
  doi          = {10.1016/0304-3975(94)90154-6}
}

@inproceedings{FiatRR90,
  author       = {Amos Fiat and
                  Yuval Rabani and
                  Yiftach Ravid},
  title        = {Competitive k-Server Algorithms (Extended Abstract)},
  booktitle    = {31st Annual Symposium on Foundations of Computer Science, {FOCS}},
  pages        = {454--463},
  year         = {1990},
  doi          = {10.1109/FSCS.1990.89566}
}

@article{FiatFKRRV98,
  author       = {Amos Fiat and
                  Dean P. Foster and
                  Howard J. Karloff and
                  Yuval Rabani and
                  Yiftach Ravid and
                  Sundar Vishwanathan},
  title        = {Competitive Algorithms for Layered Graph Traversal},
  journal      = {{SIAM} J. Comput.},
  volume       = {28},
  number       = {2},
  pages        = {447--462},
  year         = {1998},
  doi          = {10.1137/S0097539795279943}
}

@inproceedings{FreiKSW25,
  author       = {Fabian Frei and
                  Dennis Komm and
                  Moritz Stocker and
                  Philip Whittington},
  title        = {Time-Optimal k-Server},
  booktitle    = {36th International Symposium on Algorithms and Computation, {ISAAC}},
  series       = {LIPIcs},
  volume       = {359},
  pages        = {32:1--32:17},
  year         = {2025},
  doi          = {10.4230/LIPICS.ISAAC.2025.32}
}

@article{GehnenKN25,
  author       = {Matthias Gehnen and
                  Ralf Klasing and
                  {\'{E}}mile Naquin},
  title        = {Graph Exploration with Edge Weight Estimates},
  journal      = {arXiv preprint},
  year         = {2025},
  doi          = {10.48550/ARXIV.2501.18496}
}

@inproceedings{Grove91,
  author       = {Edward F. Grove},
  title        = {The Harmonic Online K-Server Algorithm Is Competitive},
  booktitle    = {Proceedings of the 23rd Annual {ACM} Symposium on Theory of Computing {STOC}},
  pages        = {260--266},
  year         = {1991},
  doi          = {10.1145/103418.103448}
}

@inproceedings{GuptaKP24,
  author       = {Anupam Gupta and
                  Amit Kumar and
                  Debmalya Panigrahi},
  title        = {Poly-logarithmic Competitiveness for the \emph{k}-Taxi Problem},
  booktitle    = {Proceedings of the 2024 {ACM-SIAM} Symposium on Discrete Algorithms, {SODA}},
  pages        = {4220--4246},
  year         = {2024},
  doi          = {10.1137/1.9781611977912.146}
}

@article{HuangZ24,
  author       = {Zhiyi Huang and
                  Hanwen Zhang},
  title        = {Deterministic 3-server on a circle and the limitation of canonical potentials},
  journal      = {Theor. Comput. Sci.},
  volume       = {1020},
  year         = {2024},
  doi          = {10.1016/J.TCS.2024.114844}
}

@article{Koutsoupias09,
  author       = {Elias Koutsoupias},
  title        = {The k-server problem},
  journal      = {Comput. Sci. Rev.},
  volume       = {3},
  number       = {2},
  pages        = {105--118},
  year         = {2009},
  doi          = {10.1016/J.COSREV.2009.04.002}
}

@inproceedings{KoutsoupiasN03,
  author       = {Elias Koutsoupias and
                  Akash Nanavati},
  title        = {The Online Matching Problem on a Line},
  booktitle    = {Approximation and Online Algorithms, First International Workshop, {WAOA}},
  pages        = {179--191},
  year         = {2003},
  doi          = {10.1007/978-3-540-24592-6\_14}
}

@article{KoutsoupiasP95,
  author       = {Elias Koutsoupias and
                  Christos H. Papadimitriou},
  title        = {On the k-Server Conjecture},
  journal      = {J. {ACM}},
  volume       = {42},
  number       = {5},
  pages        = {971--983},
  year         = {1995},
  doi          = {10.1145/210118.210128}
}

@article{KoutsoupiasP96,
  author       = {Elias Koutsoupias and
                  Christos H. Papadimitriou},
  title        = {The 2-Evader Problem},
  journal      = {Inf. Process. Lett.},
  volume       = {57},
  number       = {5},
  pages        = {249--252},
  year         = {1996},
  doi          = {10.1016/0020-0190(96)00010-5}
}

@inproceedings{Lee18,
  author       = {James R. Lee},
  editor       = {Mikkel Thorup},
  title        = {Fusible HSTs and the Randomized k-Server Conjecture},
  booktitle    = {59th {IEEE} Annual Symposium on Foundations of Computer Science, {FOCS}},
  pages        = {438--449},
  year         = {2018},
  doi          = {10.1109/FOCS.2018.00049}
}

@inproceedings{ManasseMS88,
  author       = {Mark S. Manasse and
                  Lyle A. McGeoch and
                  Daniel Dominic Sleator},
  title        = {Competitive Algorithms for On-line Problems},
  booktitle    = {Proceedings of the 20th Annual {ACM} Symposium on Theory of Computing, {STOC}},
  pages        = {322--333},
  year         = {1988},
  doi          = {10.1145/62212.62243}
}

@misc{richards_2015_22558,
  author       = {Richards, Andrew},
  title        = {{University of Oxford Advanced Research Computing}},
  month        = aug,
  year         = 2015,
  publisher    = {Zenodo},
  doi          = {10.5281/zenodo.22558},
  url          = {https://doi.org/10.5281/zenodo.22558}
}

@inproceedings{Sellke20,
  author       = {Mark Sellke},
  title        = {Chasing Convex Bodies Optimally},
  booktitle    = {Proceedings of the 2020 {ACM-SIAM} Symposium on Discrete Algorithms, {SODA}},
  pages        = {1509--1518},
  year         = {2020},
  doi          = {10.1137/1.9781611975994.92}
}

@article{Sitters14,
  author       = {Ren{\'{e}} Sitters},
  title        = {The Generalized Work Function Algorithm Is Competitive for the Generalized 2-Server Problem},
  journal      = {{SIAM} J. Comput.},
  volume       = {43},
  number       = {1},
  pages        = {96--125},
  year         = {2014},
  doi          = {10.1137/120885309}
}

@article{SittersS06,
  author       = {Ren{\'{e}} A. Sitters and
                  Leen Stougie},
  title        = {The generalized two-server problem},
  journal      = {J. {ACM}},
  volume       = {53},
  number       = {3},
  pages        = {437--458},
  year         = {2006},
  doi          = {10.1145/1147954.1147960}
}

@article{SleatorT85,
  author       = {Daniel Dominic Sleator and
                  Robert Endre Tarjan},
  title        = {Amortized Efficiency of List Update and Paging Rules},
  journal      = {Commun. {ACM}},
  volume       = {28},
  number       = {2},
  pages        = {202--208},
  year         = {1985},
  doi          = {10.1145/2786.2793}
}

@article{KelsoC82,
  author    = {Kelso, Alexander S., Jr. and Crawford, Vincent P.},
  title     = {Job Matching, Coalition Formation, and Gross Substitutes},
  journal   = {Econometrica},
  volume    = {50},
  number    = {6},
  pages     = {1483--1504},
  year      = {1982},
  publisher = {The Econometric Society},
  doi       = {10.2307/1913392}
}

@article{DressW1990,
  author    = {Dress, Andreas W. M. and Wenzel, Walter},
  title     = {Valuated matroids: A new look at the greedy algorithm},
  journal   = {Applied Mathematics Letters},
  volume    = {3},
  number    = {2},
  pages     = {33--35},
  year      = {1990},
  publisher = {Elsevier},
  doi       = {10.1016/0893-9659(90)90008-G}
}

@article{DressW1992,
  author    = {Dress, Andreas W. M. and Wenzel, Walter},
  title     = {Valuated matroids},
  journal   = {Advances in Mathematics},
  volume    = {93},
  number    = {2},
  pages     = {214--250},
  year      = {1992},
  publisher = {Elsevier},
  doi       = {10.1016/0001-8708(92)90028-J}
}

@book{Murota2003,
  author    = {Murota, Kazuo},
  title     = {Discrete Convex Analysis},
  series    = {SIAM Monographs on Discrete Mathematics and Applications},
  publisher = {Society for Industrial and Applied Mathematics},
  year      = {2003},
  doi       = {10.1137/1.9780898718508}
}

@article{SpeyerS2004,
  author    = {Speyer, David and Sturmfels, Bernd},
  title     = {The Tropical Grassmannian},
  journal   = {Advances in Geometry},
  volume    = {4},
  number    = {3},
  pages     = {389--411},
  year      = {2004},
  publisher = {De Gruyter},
  doi       = {10.1515/advg.2004.023}
}

@book{MaclaganS2015,
  author    = {Maclagan, Diane and Sturmfels, Bernd},
  title     = {Introduction to Tropical Geometry},
  series    = {Graduate Studies in Mathematics},
  volume    = {161},
  publisher = {American Mathematical Society},
  address   = {Providence, RI},
  year      = {2015},
  doi       = {10.1090/gsm/161}
}

@inproceedings{Koutsoupias99,
  author       = {Elias Koutsoupias},
  title        = {Weak Adversaries for the k-Server Problem},
  booktitle    = {40th Annual Symposium on Foundations of Computer Science, {FOCS}},
  pages        = {444--449},
  publisher    = {{IEEE} Computer Society},
  year         = {1999},
  doi          = {10.1109/SFFCS.1999.814616}
}

@book{muir1960, 
author = {Muir, Thomas}, 
title = {The Theory of Determinants in the Historical Order of Development}, 
volumes = {1-2}, 
publisher = {Dover Publications}, 
address = {New York}, 
year = {1960}, 
note = {Reprint of the 1906--1923 edition} }

@incollection{schlaefli1851,
  author    = {Schl{\"a}fli, Ludwig},
  title     = {{\"U}ber die Resultante eines Systemes mehrerer algebraischer Gleichungen},
  booktitle = {Gesammelte mathematische Abhandlungen},
  volume    = {2},
  pages     = {1--54},
  year      = {1851},
  publisher = {Birkh{\"a}user},
  address   = {Basel}
}

\appendix

\section{Omitted proofs}

\subsection*{Proof of Lemma~\ref{lem:accounting}} \label{apdx:accounting}
\begin{proof}
By the work function recurrence, we have
\begin{align*}
w_t(C_{t-1}) &= \min_{x\in C_{t-1}} w_{t-1}(X-x+r_t) + d(r_t,x)\\
&= w_{t}(C_t) + D(C_{t-1},C_t)
\end{align*}
where the last equation by definition of WFA and the fact that $w_{t-1}(C)=w_t(C)$ for all $C\ni r_t$.

Subtract $w_{t-1}(C_{t-1})$ from the equation and sum over $t$. The work function values on the right-hand side telescope, and $w_0(C_0)=0$, so
\[
w_T(C_T) + \sum_t D(C_{t-1}, C_t)
=\sum_t\bigl(w_t(C_{t-1})-w_{t-1}(C_{t-1})\bigr)
\le\sum_t \extcost_t.\qedhere
\]
\end{proof}

\subsection*{Proof of Lemma~\ref{lem:symdet}}\label{apdx:sym2proof}

\begin{proof}
Both sides multiply under composition of matrices. Gaussian elimination expresses every invertible matrix as a product of elementary changes of basis, so it suffices to check three types.

First, scaling $e_i$ by a nonzero $c$ scales $e_i^2$ by $c^2$ and each of the $k-1$ products $e_ie_j$, $j\ne i$, by $c$. It leaves all other products unchanged. The determinant in the quadratic space is therefore $c^{2+k-1}=c^{k+1}$.

Second, swapping $e_i$ and $e_j$ swaps their squares and swaps $e_ie_\ell$ with $e_je_\ell$ for each $\ell\notin\{i,j\}$. It leaves $e_ie_j$ fixed. These are $k-1$ transpositions, giving determinant $(-1)^{k-1}=(-1)^{k+1}$.

Third, consider the shear that replaces $e_i$ by $e_i+c e_j$ and fixes the other basis vectors. Order the quadratic monomials by their exponent of $e_i$. Each monomial maps to itself plus monomials having smaller exponent of $e_i$. Thus the induced matrix is triangular with all diagonal entries equal to one, and has determinant one. The original shear also has determinant one. This proves the identity for all three types and hence for their products.
\end{proof}

\end{document}